\documentclass[showpacs,amsmath,amssymb,twocolumn,pra,superscriptaddress]{revtex4-1}

\usepackage[dvips]{graphicx} 
\usepackage{amsfonts}
\usepackage{amssymb}
\usepackage{amscd}
\usepackage{amsmath}    
\usepackage{enumerate}
\usepackage{epsfig}
\usepackage{subfigure}
\usepackage{xcolor}
\usepackage{amsthm}
\usepackage{framed}
\usepackage{multirow}

\newtheorem{theorem}{Theorem}
\newtheorem{lemma}{Lemma}
\newtheorem{corollary}{Corollary}

\newtheorem{definition}{Definition}

\newtheorem{fact}{Fact}

\newcommand{\bra}[1]{\mbox{$\left\langle #1 \right|$}}
\newcommand{\ket}[1]{\mbox{$\left| #1 \right\rangle$}}

\newcommand{\comments}[1]{}
\newcommand{\CZ}{\textrm{CZ}}
\newcommand{\Tr}{\textrm{Tr}}
\begin{document}
\preprint{APS/123-QED}
\title{Constructing Multipartite Bell inequalities from stabilizers}

\date{\today}

\author{Qi Zhao}
\email{zhaoq@umd.edu}
\affiliation{Joint Center for Quantum Information and Computer Science, University of Maryland, College Park, Maryland 20742, USA}
\author{You Zhou}
\email{you\_zhou@g.harvard.edu}
\affiliation{Department of Physics, Harvard University, Cambridge, Massachusetts 02138, USA}

\begin{abstract}
Bell inequality with self-testing property has played an important role in quantum information field with both fundamental and practical applications. However, it is 
generally challenging to find Bell inequalities with self-testing property for multipartite states and actually there are not many known candidates. 
In this work, we propose a systematical framework to construct Bell inequalities from stabilizers which are maximally violated by general stabilizer states, with two observables for each local party. 
We show that the constructed Bell inequalities can 
self-test any stabilizer state which is essentially device-independent,  if and only if 
these stabilizers can uniquely determine the state in a device-dependent manner. This bridges the gap between device-independent and device-dependent verification methods. 
Our framework can provide plenty of Bell inequalities for self-testing stabilizer states. Among them, we give two families of Bell inequalities with different advantages: (1) a family of Bell inequalities with a constant ratio of quantum and classical bounds using $2N$ correlations, (2)  \emph{Single pair} inequalities improving on all previous robustness self-testing bounds using $N+1$ correlations, which are both efficient and suitable for realizations in multipartite systems. 
Our framework can not only inspire more fruitful multipartite Bell inequalities from conventional verification methods, but also pave the way for their practical applications.


\end{abstract}

\maketitle

\emph{Introduction.}---
Bell inequalities, as a test for quantum correlations, can distinguish quantum physics from its classical counterpart \cite{bell,Brunner2014nonlocality}. 
They not only play a fundamental role in quantum physics, but also demonstrate practical applications in quantum information processing, such as the quantum key distribution \cite{Mayers98,Vazirani14,Miller14,arnon2018practical}, quantum randomness generation \cite{Colbeck09,Pironio10,ma2016quantum,Acin_Certified_2016,liu2018device,bierhorst2018experimentally}, blind quantum computing \cite{reichardt2013classical,PhysRevLett.119.050503}, and quantum resource detection \cite{Moroder13,liang15,PhysRevX.7.021042}. 

One of the most striking applications of Bell inequality is the simultaneous verification of quantum states and measurements, based on the maximal violation of it \cite{Mayers98,mckague2012robust}. The verification only relies on the input and output statistics without the trust of realization of devices, which is different from traditional quantum tomography and other device-dependent methods \cite{Vogel1989Determination,Paris2004esimation}. 
This phenomenon is usually referred as self-testing. When the violation of Bell inequality is only close to its maximal quantum value, 
one can also estimate the fidelity between the underlying state and the target state which is referred as robust self-testing. 
Tremendous efforts have been made to self-test various types of quantum states \cite{wu2014robust, wu2016independent, coladangelo2017all, PhysRevLett.119.040402, baccari2018scalable}, and improve robust self-testing performance which is helpful in the realistic implementation \cite{ mckague2012robust,miller2013,kaniewski2016analytic}.



Even though Bell inequality and self-testing are of fundamental and practical significance,
it is generally challenging and not clear to propose Bell inequalities with the self-testing manner, especially for multipartite states, due to the exponentially increase of the dimension of total Hilbert space. Some interesting and inspiring attempts have been made focusing on 
high-dimensional maximally entangled states \cite{PhysRevLett.119.040402}, and graph states \cite{PhysRevA.71.042325,PhysRevLett.95.120405,PhysRevA.77.032108,baccari2018scalable}, 
nevertheless in Ref.~\cite{PhysRevA.71.042325,PhysRevLett.95.120405,PhysRevA.77.032108} the self-testing property is not explored.  Ref.~\cite{baccari2018scalable} proposed an inspiring family of Bell inequalities with self-testing manner, which are constructed from the generators of graph states. 
The ratio of the quantum and classical bounds tends to a constant in an asymptotic case. 
However, the candidates for the Bell inequality of general stabilizer states with the  self-testing manner are still limited and there is also no systematical and concrete framework for the construction of them.  



Meanwhile, 
device-dependent multipartite entanglement witnesses \cite{TERHAL2001witness,GUHNE2009detection,Friis2019Reviews} and state verification methods \cite{sam2018optimal,zhu2019efficient} 
have been extensively studied with the development of large-scale entanglement preparations \cite{Wang2018,Friis18,Gong2019genu}. 
Many efforts have been made aiming at graph states or general stabilizer states which are 
important resources in quantum information processing tasks, e.g. measurement-based quantum computing \cite{onewayQC,Raussendorf2003Measurement}, quantum routing and quantum networks \cite{Kimble2008internet,Perseguers2013network}.
The entanglement witness and the state verification of stabilizer states can be greatly simplified by utilizing the property of stabilizers \cite{Toth2005Detecting,sam2018optimal,Hayashi2015Stabilizer,Fujii2017fault,Lu2018Structure,zhou2019detecting,You2020coherent,Pei2019verification}.
Focusing on the device-independent scenario, it is thus an interesting open problem to ask whether the same properties can be applied for constructing Bell inequalities.

In this letter, we propose a systematical framework for constructing Bell inequalities with two local observables from stabilizers. We show that the necessary and sufficient condition to realize the self-testing is that a set of stabilizers used in the construction can uniquely determine the state device-dependently,
which closes the gap between the device-dependent and -independent verifications. Taking advantage of the framework, 
we provide more choices of Bell inequalities with self-testing. As applications,
two families of Bell inequalities are proposed showing different advantages. 
For any stabilizer states, we construct Bell inequalities with constant ratio of quantum and classical bounds with a linear number of correlations, 
improving the result 
in \cite{baccari2018scalable}. To further enable the robust self-testing, we 
take 3, 4-qubit GHZ states and cluster states for example. In particular, also with a linear number of correlations, we construct a new type of Bell inequalities, referred as \emph{Single pair} inequality showing the best known robust self-testing bound, which outperforms Mermin inequality \cite{mermin1990extreme,kaniewski2016analytic} and Bell inequalities proposed in \cite{baccari2018scalable}.
As a side result, our Bell inequalities can also serve as device-independent entanglement witness, which provides more alternatives for the entanglement detection in the device-dependent scenario.


\emph{Graph states and the stabilizer formalism.---}
Stabilizer states \cite{gottesman1997stabilizer,Nielsen2011Quantum} can be transformed from graph states via local unitary operations \cite{Hein2006Graph}. Thus in the following we discuss graph states without loss of generality. 

A graph state can be defined based on a graph $\mathcal{G}=(V,E)$, with vertices set $V=\{1,2,\dots,N\}$ and edges set $E\subset[V]^2$. Two vertexes $i$, $j$ are neighborhood if there is an edge $(i,j)$ connecting them, and the neighborhood set of vertex $i$ is denoted as $n_i$. Let the qubits take the role of the vertices and the edges represent the operations between the qubits, that is, the Controlled-$Z$ operation, a graph state can be written as,
\begin{equation}\label{}
\ket{\psi_\mathcal{G}}=\prod_{(i,j)\in E}\CZ^{\{i,j\}}\ket{+}^{\otimes N},
\end{equation}
where $\ket{+}=(\ket{0}+\ket{1})/\sqrt{2}$ is the eigenstate of the Pauli $X$ matrix and $\CZ^{\{i,j\}}$ is the Controlled-$Z$ gate,
$\CZ^{\{i,j\}}=\ket{0}_i\bra{0}\otimes \mathbb{I}_j+\ket{1}_i\bra{1}\otimes Z_j$.
Hereafter $X_i,Y_i,Z_i$ denote the Pauli operators of the qubit $i$.
Graph state can be uniquely determined by $N$ generators,
\begin{equation}\label{Eq:Stab}
G_i=X_i\bigotimes_{j\in n_i} Z_j,
\end{equation}
which commute with each other and staisfy $G_i\ket{\psi_\mathcal{G}}=\ket{\psi_\mathcal{G}},\, \forall i$. that is, the unique eigenstate with eigenvalue $1$ for all the $N$ generators. As a result, a graph state can also be written as a product of projectors of the generators,
\begin{equation}\label{Eq:Gsta}
\ket{\psi_\mathcal{G}}\bra{\psi_\mathcal{G}}=\prod_{i=1}^N\frac{G_i+\mathbb{I}}{2}.
\end{equation}
All the other stabilizers can be generated by the multiplication of these generators. 
The property of stabilizers can be utilized to verify the graph states and construct entanglement witness efficiently \cite{Toth2005Detecting, zhou2019detecting}. 

\emph{Constructing Bell inequalities from pairable stabilizers.}---
For a $N$-party graph state $\psi_{\mathcal{G}}$ from graph $\mathcal{G}$, we label each party with registers $1,2\dots,N$. We refer the graph $K$-colorable, if
one can label the graph with $K$ different colors requiring that there is no pair of adjacent vertices of the same color. According to this definition, we can divide all the $N$ vertices into $K$ disjoint subsets $C_{k={1,2,\cdots,K}}$, such that there is no edge inside each $C_k$. The smallest $K$ is referred as the chromatic number of $\mathcal{G}$.

Considering that each vertex $i$ owns a generator $G_i$ and the multiplication of generators from different color subsets can induce more measurement settings ($Y$ measurement). In this letter, we mainly explore the efficient Bell inequalities with two local observables for each party. Thus, we focus on the stabilizers $S$ which are generated by the generators from the same color subset $C_k$,
\begin{equation}
S=\prod_{i\in C_k} G_i.
\end{equation}
These stabilizers can be represented by tensor product of Pauli operators, $S=\bigotimes_{i}M_i$, where $M_i\in \{X,Z, \mathbb{I}\}$. For simplicity, we define a sequence $(s_1,\cdots,s_N)$ to express a stabilizer $S$ with $s_i=1$ when $M_i=X$, $s_i=-1$ when $M_i=Z$, and $s_i=0$ when $M_i=\mathbb{I}$.
Before showing the construction of Bell inequalities, let us first define a relationship between two stabilizers. 

\begin{definition}
Two stabilizers $S^1=\bigotimes_{i}M_i^1$, $S^2=\bigotimes_{i}M_i^2$ are called \emph{pairable}, if there exists at least one position $i$ such that local operators are anti-commutative, that is, $\exists~i$,  $M_i^1=Z, M_i^2=X$ or $M_i^1=X, M_i^2=Z$ $(s^1_is^2_i=-1)$.
\end{definition}
Note that \emph{pairable} stabilizers could have more than one anti-commutative position. 
Thus this definition is different with anti-commutative stabilizers. 
Hereafter, we use subscript $i$ to represent $i$th party and use superscript to represent different stabilizers. 
Taking advantage of the property of \emph{pairable} stabilizers of graph state  $\psi_{\mathcal{G}}$ to construct Bell inequalities,
we choose two \emph{pairable} stabilizers $S^1$ and $S^2$, and one anti-commutative position $T$. For position $i=T$, we replace the Pauli operators $X_i$ and $Z_i$ in the stabilizers $S^1$ and $S^2$ by observables $A_i+B_i$ and $A_i-B_i$, respectively. As for all the other positions $i \neq T$, we replace $X_i$ and $Z_i$ by $A_i$ and $B_i$, respectively. Consequently, the Bell inequality based on \emph{pairable} stabilizers $S^1$ and $S^2$ shows as follows.
\begin{lemma}\label{Th:two pair}
\begin{equation}
 \mathcal{B}_{1,2}=\sum_{j=1,2}\left \langle(A_i+s_i^jB_i)_{i=T}\prod_{i \ne T}P_i({s_i^j})\right \rangle
 \le \beta_c
\end{equation}
where $P_i(0)=\mathbb{I}$, $P_i(1)=A_i$, $P_i({-1})=B_i$, $A_i$, $B_i$ are all binary observables. The classical bound for this Bell inequality is $\beta_c= 2$ and the quantum bound (the maximal quantum value) $\beta_Q= 2\sqrt{2}$ which can be reached by the graph state $\psi_{\mathcal{G}}$.
\end{lemma}
The quantum bound is fulfilled by taking  $A_i=\frac{X+Z}{\sqrt{2}}, B_i=\frac{X-Z}{\sqrt{2}}$ when $i=T$, and $A_i=X, B_i=Z$ when $i \ne T$. 
Intuitively, we can choose a lot of $\emph{pairable}$ stabilizers to construct Bell inequalities. However, sometimes different pairs could have different anti-commutative positions which leads to inconsistent measurement settings. In order to avoid this inconsistency, we define a vertex subset $\mathcal{AC}\subset V$ to decide the observables of some specific positions, such as $\mathcal{AC}=\{T\}$ in  $\mathcal{B}_{1,2}$. We want to combine a lot of $\emph{pairable}$ stabilizers, and define the set $\mathcal{P}$ containing all the chosen pairs $S^l$ and $S^k$ as $(l, k)\in \mathcal{P}$. Note that for different pairs, they can share the same stabilizer, for example we allow $(1,2),(1,3)\in \mathcal{P}$.  Besides the $\emph{pairable}$ stabilizers, one can also add other non-pairable stabilizers $S^r$, $r\in \mathcal{R} $. 
We define the set containing all the stabilizers appearing in $\mathcal{P}$ and $\mathcal{R}$ as $\mathcal{ST}$ .
For given $\mathcal{P}$, $\mathcal{R}$, we put some requirements on $\mathcal{AC}$ as follows.
\begin{enumerate}
    \item For every $\emph{pairable}$ stabilizers, $S^l$ and $S^k$ , $\{l, k\}\in \mathcal{P} $, there exists only one position $T_{l,k}\in \mathcal{AC}$ such that the measurement setting of stabilizers, $S^l$ and $S^k$ in this position are anti-commutative.
    For other positions $i\in \mathcal{AC} \verb|\| \{T_{l,k}\}$, the measurement settings in $S^l$ and $S^k$ are all $\mathbb{I}$, $s^l_i=s^k_i=0$.
    \item For every stabilizer $S^r (r\in \mathcal{R}) $, the measurements in any position $i\in \mathcal{AC}$ are $\mathbb{I}$, $s^r_i=0$.  
\end{enumerate}

These requirements are used to prevent the situation where the maximal classical value of Bell correlation exceeds its maximal quantum value. 
Though this is not strictly forbidden, it would reduce the ratio of quantum and classical bounds in the constructed Bell inequalities.
Luckily, one can always find a suitable $\mathcal{AC}$ for the chosen stabilizer set $\mathcal{ST}$ by the following lemma.

\begin{lemma}\label{Le:suitable}
For any given stabilizer set $\mathcal{ST}$ containing \emph{pairable} stabilizers, one can always assign the stabilizers of $\mathcal{ST}$ into the paring set $\mathcal{P}$ and the remaining subset $\mathcal{R}$, and then find a suitable $\mathcal{AC}$ satisfying two requirements listed above.
\end{lemma}
Thus it is generally feasible to construct Bell inequalities from a stabilizer set $\mathcal{ST}$.
In order to make this construction clearer, we show an example of 4-qubit cluster state in FIG.~\ref{clusterex}. 
There are also other constructions for this given stabilizers set $\mathcal{ST}$, here we only give two constructions as examples.

\begin{figure}[htbp]
\centering
\includegraphics[width=0.35\textwidth]{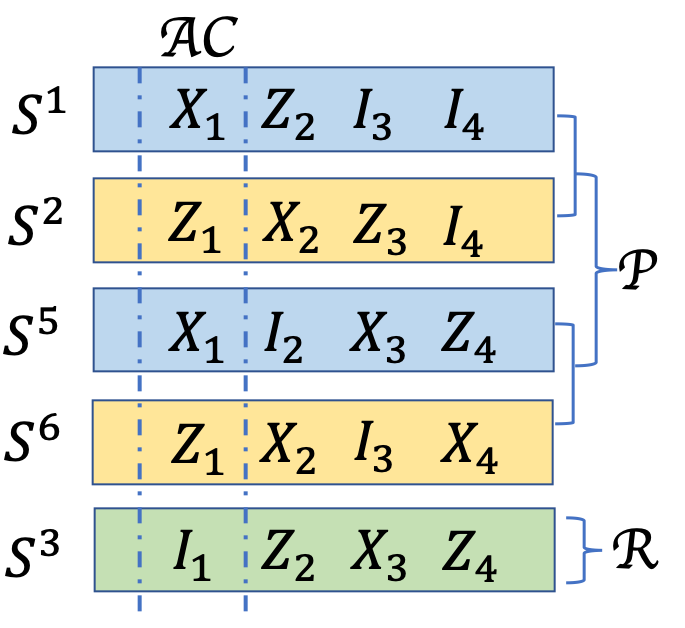}
\caption{An example of $\mathcal{P}, \mathcal{R}$ with a suitable $\mathcal{AC}$ for the 4-party cluster state.  For a given $\mathcal{ST}=\{1,2,3,5,6\}$ where $S^1=G_1=X_1Z_2$,$S^2=G_2=Z_1X_2Z_3$, $S^3=G_3=Z_2X_3Z_4$,$S^5=G_1G_3$, $S^6=G_2G_4$, we can assign $\mathcal{P}=\{(1,2),(5,6)\}$, $\mathcal{R}=\{3\}$ with $\mathcal{AC}=\{1\}$.
We can also assign $\mathcal{P}=\{(2,3),(2,5)\}$, $\mathcal{R}=\{1,6\}$ with $\mathcal{AC}=\{3\}$.
}\label{clusterex}
\end{figure}

For $\mathcal{P}$, $\mathcal{R}$ with a suitable $\mathcal{AC}$ satisfying the above two requirements, we construct the Bell inequalities for general graph states as follows.   
\begin{theorem}\label{Th:construction}
\begin{equation}
  \begin{aligned}
  &\sum_{(l,k)\in \mathcal{P}}\mathcal{B}_{l,k}+ \sum_{r\in \mathcal{R}} \mathcal{B}_{r}\le \beta_c=2|\mathcal{P}|+ |\mathcal{R}|,\\
&\mathcal{B}_{l,k}=\sum_{j=l,k}\left \langle\prod_{i\in \mathcal{AC}}(A_i+s^j_iB_i)\prod_{i\notin \mathcal{AC}}P_i({s_i^j})\right \rangle,\\
&\mathcal{B}_{r}=\left \langle\prod_{i\notin \mathcal{AC}}P_i({s_i^r})\right \rangle,
\end{aligned}
\end{equation}
where $P_i({s_i^j})$ are defined in Lemma~\ref{Th:two pair}, $|\mathcal{P}|$ and $|\mathcal{R}|$ denote the number of elements in sets $\mathcal{P}$ and $\mathcal{R}$, respectively. The quantum bound $\beta_Q=2\sqrt{2}|\mathcal{P}|+ |\mathcal{R}|$ and the corresponding graph state $\psi_{\mathcal{G}}$ can reach this maximal quantum value.
\end{theorem}
Note that novel CHSH-like multipartite Bell inequalities proposed in \cite{baccari2018scalable} can be seen as special cases of Theorem \ref{Th:construction} by  choosing all the stabilizers being generators in Eq.~\eqref{Eq:Stab}, that is, $\mathcal{ST}= \{S^i|S^i=G_i, i\in \{1,\dots,N\}\}$, where $G_i$ is the generator asscociated with $i$th vertex. In particular,
assuming the first vertex is the one with the largest number of neighbours, $|n_1|=\max_i |n_i|$,  the generators in $\mathcal{ST}$ are assigned into $\mathcal{P}=\{(1,j)|j\in n_1\}$, $\mathcal{R}=\{r|r\ne n_1\cup 1 \}$
with $\mathcal{AC}=\{1\}$. 


\emph{Sufficient and necessary condition for self-testing.}---
Besides ruling out the classical hidden variable model, Bell inequalities further provide us a method to verify quantum states in a device-independent manner. 
Though graph state $\psi_{\mathcal{G}}$ can reach the maximal quantum value of all Bell inequalities constructed from its stabilizers in Theorem \ref{Th:construction}, not all these Bell inequalities can verify $\psi_{\mathcal{G}}$ uniquely. 
Here in this section, based on the constructed Bell inequalities, we explore the sufficient and necessary condition for the self-testing of graph states. 
\begin{definition}
Suppose that the Bell inequality $\mathcal{B}_\mathcal{G}$, constructed from the stabilizers of $\psi_{\mathcal{G}}$, 
is maximally violated by a state $\psi$ and corresponding local observables $A_i$, $B_i$. If up to local isometries, $\psi$ and these local observables
are equivalent to the graph state $\psi_{\mathcal{G}}$, and $\frac{X_i+Z_i}{\sqrt{2}}, \frac{X_i-Z_i}{\sqrt{2}}$
when $i\in \mathcal{AC}$; $X_i, Z_i$ when $i\notin \mathcal{AC}$, we say this Bell inequality $\mathcal{B}_\mathcal{G}$ can self-test graph state $\psi_{\mathcal{G}}$. 
\end{definition}

In the conventional device-dependent verification of graph states, we have the following fact about determining graph state $\psi_{\mathcal{G}}$ with the trust of measurement devices.
\begin{fact}
\cite{gottesman1997stabilizer,Nielsen2011Quantum} For the device-dependent verification, a set of stabilizer measurements can uniquely determine the state $\psi_{\mathcal{G}}$, if and only if it contains $N$ independent stabilizers. A set a stabilizers are independent if they can not generate each other by multiplication, and note that any other stabilizers can be generated by $N$ independent stabilizers.

\end{fact}

\begin{theorem}\label{Th:self-testing}
The family of Bell inequalities proposed in Theorem \ref{Th:construction} can self-test the graph state $\psi_{\mathcal{G}}$, if and only if the stabilizers in $\mathcal{ST}$ together can determine the graph state $\psi_{\mathcal{G}}$. 
\end{theorem}
Via this theorem, we make a close connection between the device-independent and -dependent verifications. 
If the stabilizers can determine the graph state under the trust of measurement devices, then one can always transform these stabilizers to Bell inequalities and apply it to verify the state without trusting the devices under our framework. Thus any witnesses and state verification methods based on stabilizers could inspire the construction of Bell inequalities.



\emph{Applications.---}
Equipped with the above framework, we exhibit applications via constructing Bell inequalities with various advantages. Note that in this section we only explore Bell inequalities with the self-testing property. 

At first, we prefer to select Bell inequalities with a large ratio of quantum and classical bound $\beta_Q/\beta_C$, which is beneficial for the experimental violation under practical set-ups and can lead to good performance in cryptography tasks \cite{Miller14}. 
From Theorem ~\ref{Th:construction}, it is clear that the maximal $\beta_Q/\beta_C$ is $\sqrt{2}$ when $\mathcal{R}=\varnothing$ in our construction. 
Utilizing the properly designed multiplications of generators from the graph state, we can construct self-testing Bell inequalities with the maximal ratio for any graph state.
\begin{corollary}\label{Cor:constant}
For any $N$-party graph state $\psi_{\mathcal{G}}$, based on our framework one can construct Bell inequalities using $2N$ correlations  to reach the maximal 
$\beta_Q/\beta_C=\sqrt{2}$ with self-testing
property in the same time.
\end{corollary}
The proof and detailed constructions are shown in Appendix C. 
In Ref.~\cite{baccari2018scalable}, only asymptotic case (infinity large $N$) and some special states, for example, GHZ state can reach this ratio $\beta_Q/\beta_C=\sqrt{2}$.  
Due to flexible choices of stabilizers, instead of only using generators, our framework provides Bell inequalities with generally larger $\beta_Q/\beta_C$ than the constructions in \cite{baccari2018scalable}, also using linear number of correlations, which are efficient and scalable for practical demonstrations in multipartite system.

Secondly, we also expect that Bell inequalities can show a good performance in the robust self-testing task. In the robust self-testing, one would like to lower bound the state fidelity to the target graph state $\psi_{\mathcal{G}}$ (under local isometries), only on account of the Bell inequality value which deviates from the maximal quantum value.  
Benefiting from the flexibility of our construction, one can have many choices of Bell inequalities at hand. We give constructions of some typical examples, 3-qubit and 4-qubit GHZ states, and 4-qubit 1-D cluster state in Table \ref{T:3GHZ}, \ref{T:4GHZ} and \ref{T:cluster}, respectively, with different $\mathcal{AC},\mathcal{P}, \mathcal{R}$ and $\beta_Q/\beta_C$. 
Based on the method in \cite{baccari2018scalable,kaniewski2016analytic}, in FIG.~\ref{Fig:compare}, we numerically show the performance of them in robust self-testing, and also compare with the Mermin inequality which is widely used for self-testing GHZ state \cite{mermin1990extreme,kaniewski2016analytic}. In principle, the robustness analysis can be extended to more qubits and one can explore more possible Bell inequalities via our framework.
\begin{table}[ht]
\begin{tabular}{|cl|cccc|}
  
  \hline
  &$\mathrm{GHZ}_3$ ($\mathrm{Cluster}_3$) & $\mathcal{AC}$ & $\mathcal{P}$ & $\mathcal{R}$  & $\beta_Q/\beta_C$\\[1mm]
  \hline
 & $\mathbf{1}$  & $\{1\}$
                        & $\{(1,2)\}$
                             & $\{3\}$ & $(2\sqrt{2}+1)/3$  \\[1mm]
      \hline
    & $\mathbf{2}$    &$\{1\}$
                        & $\{(1,2), (2,4)\}$
                             & $\{3\}$
                                 & $(4\sqrt{2}+1)/5$  \\[1mm]                        
                             
  \hline
    & $\mathbf{3}$    &$\{1\}$
                        & $\{(1,2), (2,4)\}$
                             & $\varnothing$
                                 & $\sqrt{2}$  \\[1mm]
  \hline
    & 4 \cite{baccari2018scalable} & $\{2\}$
                        & $\{(1,2), (2,3)\}$
                            & $\varnothing$ & $\sqrt{2}$ \\[1mm]

  \hline
\end{tabular}
\caption{3-qubit GHZ (Cluster) state with with generators $G_1=X_1Z_2$, $G_2=Z_1X_2Z_3$ and $G_3=Z_2X_3$.
And we choose stabilizers as $S^1=G_1$, $S^2=G_2$, $S^3=G_3$ and $S^4=G_1G_3$. According to our framework, actually we can construct totally 12 Bell inequalities with self-testing. Here we only take few of them as examples. Our new constructions are shown in bold.}
\label{T:3GHZ}
\end{table}

\begin{table}[ht]
\begin{tabular}{|cl|cccc|}
  
  \hline
  &$\mathrm{GHZ}_4$ & $\mathcal{AC}$ & $\mathcal{P}$ & $\mathcal{R}$  & $\beta_Q/\beta_C$\\[1mm]
   \hline
    & $\mathbf{1}$   &$\{1\}$
                        & $\{(1,2)\}$
                             & $\{5,6\}$
                                 & $(\sqrt{2}+1)/2$   \\[1mm]

  \hline
    & $\mathbf{2}$ & $\{2\}$
                        & $\{(1,2)\}$
                            & $\{3,4\}$ & $(\sqrt{2}+1)/2$  \\[1mm]

 \hline
    & $\mathbf{3}$& $\{1\}$
                        & $\{(1,2), (1,3)\}$
                            & \{6\} & $(4\sqrt{2}+1)/5$  \\[1mm]
 \hline
 & 4 \cite{baccari2018scalable}  & $\{1\}$
                        & $\{(1,2), (1,3), (1,4)\}$
                             &  $\varnothing$ & $\sqrt{2}$  \\[1mm]
 \hline
\end{tabular}
\caption{4-qubit GHZ state with generators  
$G_1=X_1Z_2Z_3Z_4$, $G_2=Z_1X_2$,  $G_3=Z_1X_3$ and $G_4=Z_1X_4$. And we choose stabilizers as $S^1=G_1$, $S^2=G_2$, $S^3=G_3$, $S^4=G_4$, $S^5=G_2G_3$ and $S^6= G_2G_4$.  Our new constructions are shown in bold.}
\label{T:4GHZ}
\end{table}

\begin{table}[ht]
\begin{tabular}{|cl|cccc|}
  
  \hline
  &$\mathrm{Cluster}_4$ & $\mathcal{AC}$ & $\mathcal{P}$ & $\mathcal{R}$  & $\beta_Q/\beta_C$\\[1mm]
  \hline
 & $\mathbf{1}$ & $\{1\}$
                        & $\{(1,2)\}$
                             & $\{3,4\}$ & $(\sqrt{2}+1)/2$  \\[1mm]
 \hline
    & 2  \cite{baccari2018scalable}& $\{2\}$
                        & $\{(1,2), (2,3)\}$
                            & \{4\} & $(4\sqrt{2}+1)/5$  \\[1mm]

  \hline
    & $\mathbf{3}$   &$\{1\}$
                        & $\{(1,2), (5,6)\}$
                             & $\varnothing$
                                 & $\sqrt{2}$  \\[1mm]
 
  \hline
    & $\mathbf{4}$ & $\{2\}$
                        & $\{(1,2), (3,6)\}$
                            & $\varnothing$ & $\sqrt{2}$ \\[1mm]

  \hline
\end{tabular}
\caption{4-qubit 1-D cluster state with generators $G_1=X_1Z_2$, $G_2=Z_1X_2Z_3$, $G_3=Z_2X_3Z_4$, $G_4=Z_3X_4$. And we choose stabilizers as
$S^1=G_1$, $S^2=G_2$, $S^3=G_3$, $S^4=G_4$, $S^5=G_1G_3$, $S^6= G_2G_4$.  Our new constructions are shown in bold.}
\label{T:cluster}
\end{table}


\begin{figure*}[htbp]
\centering
\includegraphics[width=1\textwidth]{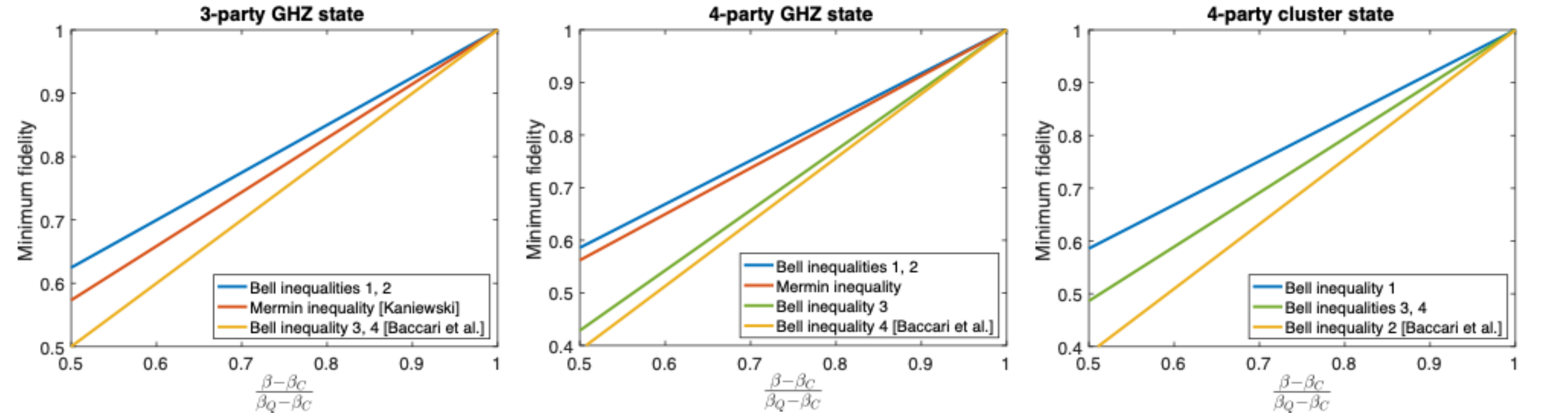}
\caption{Numerical estimation of the lower bound of fidelity to the target graph state versus normalized Bell inequality value $\frac{\beta-\beta_C}{\beta_Q-\beta_C}$. We compare different constructions of Bell inequalities from our framework with previous results. The best candidates for robust self-testing are \emph{Single pair} inequalities, which are the blue curves in the figure, where it shows the best known robustness self-testing performance, surpassing Mermin inequality in \cite{kaniewski2016analytic} 
and Bell inequalities proposed in  \cite{baccari2018scalable},
which are shown in red curve and yellow curve, respectively. 
}\label{Fig:compare}
\end{figure*}

From FIG.~\ref{Fig:compare}, one can find that larger $\beta_Q/\beta_C$ value does not necessarily lead to a better performance in the robust self-testing, which is characterized by the slope of the curves. 
The reason behind this phenomenon may be that 
the method of estimating fidelity is not tight or optimal. In our examples, we find that the best robust self-testing Bell inequality share the same property, that is, they are constructed from only one \emph{pairable} stabilizers, i.e., $|\mathcal{P}|=1$, and other stabilizers all in $\mathcal{R}$ (See Bell inequality 1 in Table \ref{T:3GHZ}, Bell inequalities 1 and 2 in Table \ref{T:4GHZ}, Bell inequality 1 in Table \ref{T:cluster}).  We name them as \emph{Single~pair} Bell inequalities,  \begin{equation}
 \mathcal{B}_{1,2}+\sum_{r\in \mathcal{R}}\mathcal{B}_{r}\le 2+|\mathcal{R}|.
\end{equation}

From these examples and the
realizable lower bounds of the fidelity, one can see that our new constructed \emph{Single~pair} Bell inequalities provide the best known robust self-testing bound, improving the previous self-test bound in \cite{kaniewski2016analytic,baccari2018scalable}. Note that compared to Mermin inequality with $2^N$ correlations, \emph{Single~pair} Bell inequalities only contains $N+1$ correlations, which is more efficient and scalable for large-scale system verification. This shows the potentiality of our framework.


As a side result of the robust self-testing bound, one can also construct device-independent genuine entanglement witness, by applying the linear self-testing fidelity bound $F \ge a\frac{\beta-\beta_C}{\beta_Q-\beta_C}+b $ with slope $a$ and intercept $b$,
\begin{corollary}
\begin{equation}
  \begin{aligned}
  &\sum_{(l,k)\in \mathcal{P}}\mathcal{B}_{l,k}+ \sum_{r\in \mathcal{R}} \mathcal{B}_{r}\stackrel{\text{bi-sep.}}{\le} \beta_{0.5},\\
  &\beta_{0.5}= \frac{(0.5-b)(\beta_Q-\beta_C)}{a}+\beta_C
\end{aligned}
\end{equation}
where $\mathcal{B}_{i}$ and $\mathcal{B}_{r}$ are Bell correlations shown in Theorem \ref{Th:construction}, $\beta_{0.5}$ is the threshold Bell value. The violation of this inequality implies the existence of genuine entanglement. \end{corollary}
This corollary is due to the fact that the underlying state possesses genuine entanglement when fidelity with a certain graph state exceeds 0.5 \cite{Toth2005Detecting, zhou2019detecting}.
We give detailed genuine entanglement bounds $\beta_{0.5}$ in Appendix E for our constructed Bell inequalities. 

\emph{Conclusion}---
In summary, we propose a systematical framework to construct Bell inequalities directly from stabilizers, and further provide a one-to-one map from the device-dependent verification to the self-testing one. 

The framework can provides us a large number of Bell inequalities to select for different application scenarios, for instance, \emph{Single~pair} Bell inequality for robust self-testing.
Even though the fidelity lower bounds for robust self-test are obtained by numerics, 
these result are also instructive for obtaining an (tight) analytical bound for general graph states in the future \cite{kaniewski2016analytic}.
One may also find other interesting inequalities with other advantages. 
Similar to the entanglement witness, a modification of coefficients between different Bell expressions may be beneficial for improving the ratio of  quantum and classical bounds and the robust self-testing performance \cite{Lu2018Structure,zhou2019detecting,zhao2019efficient}.  

Via the proposed framework, we close the gap between device-independent and -dependent witness and verification, and borrow the experience from device-dependent study. This connection can inspire more complicated Bell inequalities constructions, for instance, the ones with more than two local measurement, applying $Y$ operators in stabilizers.  
It would be also interesting to extend the method to hypergraph state, non-stabilizer states, and high-dimensional entangled states \cite{Rossi2013hyper,Kraft2018Multilevel,Bavaresco2018high,You2019Sym}. 

Under the context of entanglement detection, like the entanglement witnesses shown above, one can also obtain device-independent entanglement (structure) witnesses \cite{liang15,Lu2018Structure,zhou2019detecting}, and we leave the details in the future work.  

\emph{Acknowledgement}
We are grateful to Xiao Yuan for useful discussions. 
QZ acknowledges the support by the Department of Defense through the Hartree Postdoctoral Fellowship at QuICS.
YZ was supported in part by the Templeton Religion Trust under grant TRT 0159.

%

\begin{appendix}

\section{Proof for Constructed Bell inequalities}

\emph{Proof for Lemma \ref{Th:two pair}}---. For any \emph{pairable} stabilizers $S^1$ and $S^2$, we choose one of its anti-commutative position as $T$, such that $s_T^1s^2_T=-1$. It is not hard to see that the classical value of $\mathcal{B}_{1,2}$ can not be greater than $\left \langle A_T+B_T\right \rangle+\left \langle A_T-B_T\right \rangle\le 2$.

Then we prove the quantum bound by expressing the difference as \begin{widetext}
     \begin{eqnarray}\label{Eq:maximal pair1}
      2\sqrt{2}- \mathcal{B}_{1,2}=\left \langle \frac{1}{\sqrt{2}}\left[\mathbb{I}-(\frac{A_T+ B_T}{\sqrt{2}})\prod_{i\ne T}P_k({s_i^1})\right]^2+ \frac{1}{\sqrt{2}}\left[\mathbb{I}-(\frac{A_T-B_T}{\sqrt{2}})\prod_{i\ne T}P_k({s_i^2})\right]^2 \right \rangle\ge 0.
     \end{eqnarray}
\end{widetext}
The quantum bound is reached by choosing  $A_T=\frac{X+Z}{\sqrt{2}}, B_T=\frac{X-Z}{\sqrt{2}}$
and $A_i=X, B_i=Z$ for $i\ne T$ with the underlying state is the graph state $\psi_{\mathcal{G}}$.

\emph{Proof for Lemma \ref{Le:suitable}}.---
We choose one pair of \emph{pairable} stabilizers $S^1$ and $S^2$, and choose one anti-commutative position, denoted as $T$. We divide all the stabilizers belonging to $\mathcal{ST}$ into three subsets according to the measurement setting on the $T$th position,
\begin{equation}
\begin{aligned}
    &\mathcal{P}_+=\{S|s_T=1,S\in \mathcal{ST}\},\\
    &\mathcal{P}_-=\{S|s_T=-1,S\in \mathcal{ST}\},\\
     &\mathcal{P}_0=\{S|s_T=0,S\in \mathcal{ST}\}.
\end{aligned}
\end{equation}
Any stabilizers $S^l\in \mathcal{P}_+$ and 
$S^k\in \mathcal{P}_-$ are \emph{pairable}. If $|\mathcal{P}_+|\ne |\mathcal{P}_-|$, we can reuse the stabilizers. In this way, all the stabilizers in  $\mathcal{P}_+\cup \mathcal{P}_-$ can be assigned into paring set $\mathcal{P}$ and also let $\mathcal{P}_0=\mathcal{R}$ and $\mathcal{AC}=T$. The above construction satisfies two listed requirements.

\emph{Proof for Theorem \ref{Th:two pair}.---}
First, the classical bound can be obtained by multiple uses of the result in Lemma 1, 
that is, $\sum_{(l,k)\in \mathcal{P}}(\mathcal{B}_{l}+\mathcal{B}_{k})\le 2|\mathcal{P}| $ and the fact $\sum_{r\in \mathcal{R}}\mathcal{B}_{r}\le |\mathcal{R}|$. 

Then we prove the quantum bound by showing that  
  \begin{equation}
      2\sqrt{2}|\mathcal{P}|+ |\mathcal{R}|- \sum_{(l,k)\in \mathcal{P}}(\mathcal{B}_{l}+\mathcal{B}_{k})- \sum_{r\in \mathcal{R}} \mathcal{B}_{r}\ge 0. 
  \end{equation}
Note that here in each $\mathcal{B}_{l}$, $\mathcal{B}_{k}$ and $\mathcal{B}_{r}$ all the observables $A$ and $B$ are not fixed measurement settings, they are arbitrary dichotomic observables. Similar as the proof of Lemma \ref{Th:two pair},
this inequality can be transformed into a sum of the following squares. 
\begin{widetext}
     \begin{eqnarray}\label{Eq:maximal pair}\label{Eq:Sq:P}
     2\sqrt{2}-\mathcal{B}_{l}-\mathcal{B}_{k}=\left\langle\frac{1}{\sqrt{2}}\left[\mathbb{I}-\prod_{i\in \mathcal{AC}}(A_i+s^l_iB_i)\prod_{i\notin \mathcal{AC}}P_i({s_i^l})\right]^2+ \frac{1}{\sqrt{2}}\left[\mathbb{I}-\prod_{i\in \mathcal{AC}}(A_i+s^k_iB_i)\prod_{i\notin \mathcal{AC}}P_i({s_i^k})\right]^2\right\rangle\ge 0.
     \end{eqnarray}
\end{widetext}
\begin{equation}\label{Eq:maximal single}
    \begin{aligned}
 1-\mathcal{B}_{r}=\left \langle\frac{1}{2}\left[\mathbb{I}-\prod_{i\notin \mathcal{AC}}P_k({s_i^j})\right]^2 \right\rangle\ge 0.
 \end{aligned}
\end{equation}
The quantum bound is reached by choosing  $A_i=\frac{X+Z}{\sqrt{2}}, B_i=\frac{X-Z}{\sqrt{2}}$
when $i\in \mathcal{AC}$. $A_i=X, A_i=Z$ when $i\notin \mathcal{AC}$ with the underlying state is the graph state $\psi_{\mathcal{G}}$.

\section{Proof for Self-testing}
 
\emph{Proof for Theorem \ref{Th:self-testing}}.---
At first, we prove the ``if " part. We assume the underlying state is $\psi$ and  observables are $O_i$ for the $i$th party, and prove that for all correlations in Bell inequality there exists an isometry $\Phi=\bigotimes_{i=1}^{N} \Phi_i$, with local isometries $\Phi_i$,
\begin{equation}
\Phi\left[ \prod_i O_i(s^j_i) \ket{\psi}\right]= \prod_{i}P_i({s_i^j}) ~\left(\ket{\psi_{\mathcal{G}}} \otimes \ket{extra}\right)
\end{equation}
where for $i\in \mathcal{AC}$, $P_i(0)=\mathbb{I}$, $P_i(1)=\frac{X+Z}{\sqrt{2}}$; $P_i({-1})=\frac{X-Z}{\sqrt{2}}$, for $i\notin \mathcal{AC}$,
$P_i(0)=\mathbb{I}, P_i(1)=X$, $P_i({-1})=Z$.

Without loss of generality, in the following we assume $\mathcal{AC}=\{1\}$ and the case where $\mathcal{AC}$ contains more than one position can be proved similarly. Denote 
\begin{equation}\label{Eq: defX1}
\tilde{X_1}=\frac{1}{c\sqrt{2}}\left[O_1(1)+O_1(-1)\right], \tilde{Z_1}=\frac{1}{c'\sqrt{2}}[O_1(1)-O_1(-1)],
\end{equation}
where $c, c'$ are normalized parameters such that $|\tilde{X_1}|=|\tilde{Z_1}|=1$. For $i\ne 1$, we denote 
$\tilde{X_i}=O_i(1)$ and $\tilde{Z_i}=O_i(-1)$ which are the actually implemented observables. 

A key step of the proof is to show that for each position $i$, when acting on the underlying state $\psi$, $\tilde{X_i}$ and $\tilde{Z_i}$ are anti-commutative, i.e.,
\begin{equation}\label{Eq: anti}
(\tilde{X_i}\tilde{Z_i}+\tilde{Z_i}\tilde{X_i})\ket{\psi}=0, ~~i=1,2\dots, N.
\end{equation}

When the Bell inequality reach the maximal value, the inequalities in Eq.~\eqref{Eq:Sq:P} and \eqref{Eq:maximal single} in the proof of Theorem 1 should be saturated when acting on the state $\psi$. As a result, according to each square in the equations, we have 
\begin{equation}
\begin{aligned}\label{Satu1}
\tilde{X_1}\ket{\psi}&= \prod_{i\ne 1}O_i({s_i^l}) \ket{\psi},\\
\tilde{Z_1}\ket{\psi}&= \prod_{i\ne 1}O_i({s_i^k}) \ket{\psi},\\
\ket{\psi}&= \prod_{i\ne 1}O_i({s_i^r}) \ket{\psi},\\
\end{aligned}
\end{equation}
where $(l,k)\in \mathcal{P}$ and  $r\in \mathcal{R}$.
For the position $i=1$, the anticommutive relationship can be obtained directly from  Eq.~\eqref{Eq: defX1},
\begin{equation}\label{Eq: anti1}
(\tilde{X_1}\tilde{Z_1}+\tilde{Z_1}\tilde{X_1})\ket{\psi}=0.
\end{equation}
Since the stabilizers in $\mathcal{ST}= \mathcal{P} \cup \mathcal{R}$ can determine the graph state $\psi_\mathcal{G}$, $\mathcal{ST}$ at least contains $N$ independent stabilizers. For all generators in Eq.~\eqref{Eq:Stab}, there always exist a series of stabilizers $S^{i_1},\dots,S^{i_m}$ from $\mathcal{ST}$ satisfying
\begin{equation}\label{Eq:Stabin}
G_i=\prod_{j=1}^{m} S^{i_j},
\end{equation} 
with $m\leq N$.
By multiple uses of Eq.~\eqref{Satu1} and plugging into Eq.~\eqref{Eq:Stabin}, we have 
\begin{equation}\label{Eq:Ap:stab}
\begin{aligned}
&\tilde{X_i}\ket{\psi}= \prod_{j\in n_i} \tilde{Z_j}\ket{\psi}.
\end{aligned}
\end{equation}
for all $1\leq i\leq N$. Thus for $j\in n_1$, by utilizing the relations of vertex $1$ and $j$ in Eq.~\eqref{Eq:Ap:stab}, we have
\begin{widetext}
     \begin{eqnarray}
(\tilde{Z_j}\tilde{X_j}+\tilde{X_j}\tilde{Z_j})\ket{\psi}= (\tilde{Z_1}\tilde{X_1}+\tilde{X_1}\tilde{Z_1})  \prod_{i\in \mathcal{C}(1,j) }Z_i \ket{\psi},
\end{eqnarray}
\end{widetext}
where $\mathcal{C}(1,j)=[n_1\cup n_j] \setminus [\{1,j\}\cup (n_1\cap n_j)]$ denotes the set of all the vertexes which are neighbors of either $1$ or $j$, but are not $1$ or $j$ themselves.
Starting from position 1, we can obtain the anti-commutative relationship for the positions $j\in n_1$,
\begin{equation}
(\tilde{X_j}\tilde{Z_j}+\tilde{Z_j}\tilde{X_j})\ket{\psi}=0, ~~j\in n_1.
\end{equation}

Then in the same way, we can to obtain the anti-commutative relationship for the operators whose corresponding vertex is the neighbor of the vertexes in $n(1)$. Because the graph is connected, we can iterate the above procedure and get the anti-commutative relationship for all parties. The construction of the isometry $\Phi$ is exactly the same with that in \cite{baccari2018scalable} and we do not repeat it here. 

Secondly, we prove the ``only if " part. We assume that there exists 
a stabilizer set $\mathcal{ST}$, 
which can not determine the graph state. That is, it can at most contain $N-1$ independent stabilizers denoted by $S^1, S^2,\cdots, S^{N-1}$, and we can always find one generator denoted by $S^N=G_N$ which can not be expressed by the product of stabilizers in $\mathcal{ST}$. Consequently, we construct a state
\begin{equation}
\begin{aligned}
\rho=\frac1{2}\prod_{i=1}^{N-1}\frac{S^i+\mathbb{I}}{2},
\end{aligned}
\end{equation}
which is the maximally mixed state in the two dimensional subspace, determined by $S^1, S^2,\cdots, S^{N-1}$ all taking the eigenvalue 1. As a result, $\rho$ has exactly the same value with $\psi_{\mathcal{G}}=\prod_{i=1}^N\frac{S^i+\mathbb{I}}{2}$ for Bell inequalities constructed from $\mathcal{ST}$. Note that $\rho$ is actually the mixture of $\psi_{\mathcal{G}}$ and another state,
\begin{equation}
\begin{aligned}
\rho&=\frac{(\mathbb{I}+G_N)/2+(\mathbb{I}-G_N)/2}{2}\prod_{i=1}^{N-1}\frac{S^i+\mathbb{I}}{2}\\
&=\frac1{2}(\psi_{\mathcal{G}}+\psi'_{\mathcal{G}}).
\end{aligned}
\end{equation}
Here the state $\psi'_{\mathcal{G}}$ is determined by $S_1, S_2,\cdots, S_{N-1}$ all taking the eigenvalue 1 and $G_N$ taking $-1$,
\begin{equation}\label{Eq:Ap:LU}
\begin{aligned}
\psi'_{\mathcal{G}}=\frac{(\mathbb{I}-G_N)}{2}\prod_{i=1}^{N-1}\frac{S_i+\mathbb{I}}{2}.
\end{aligned}
\end{equation}
Since $\{S_i\}_1^N$ is complete, one can uniquely determine whether $G_i\ket{\psi'_{\mathcal{G}}}=\ket{\psi'_{\mathcal{G}}}$ or $-\ket{\psi'_{\mathcal{G}}}$, by multiplying the results of them. And we denote the vertex subset $\mathcal{D}=\left\{i:G_i\ket{\psi'_{\mathcal{G}}}=-\ket{\psi'_{\mathcal{G}}}\right\}$. As a result, $\psi'_{\mathcal{G}}$ can also be transformed from $\psi_{\mathcal{G}}$ by local unitary,
\begin{equation}
\begin{aligned}
\ket{\psi'_{\mathcal{G}}}=\prod_{i\in \mathcal{D}}Z_i\ket{\psi_{\mathcal{G}}}.
\end{aligned}
\end{equation}

In the following, we show that $\rho$ can not be transformed to $\psi_{\mathcal{G}}$ by local isometries, which contradicts to the self-testing claim. Let us focus on any single qubit, say the 1st qubit, and take it as the subsystem $B$, and the remaining qubits as $A$. The quantum conditional entropy on $A$ of the state $\psi_{\mathcal{G}}$ is
\begin{equation}
\begin{aligned}
S(A|B)_{\psi_\mathcal{G}}=S(AB)_{\psi_\mathcal{G}}-S(B)_{\psi_\mathcal{G}}=-1\\
\end{aligned}
\end{equation}
where we use the fact that $\psi_\mathcal{G}$ is pure and the entanglement entropy of the first qubit is $1$ \cite{Hein2006Graph}. On the other hand, quantum conditional entropy of the state $\rho$ shows
\begin{equation}
\begin{aligned}
S(A|B)_{\rho}&=S(AB)_{\rho}-S(B)_{\rho}\\
&=1-S(B)_{\rho}\geq 0,
\end{aligned}
\end{equation}
where we apply the fact that $\rho$ is a maximally mixed state in the subspace, and the entropy of $B$ is upper bounded by the qubit number $1$. 

It is known that $S(A|B)$ quantifies how many qubits need to send from $A$ (Alice) to $B$ (Bob) to reconstruct $\rho_{AB}$ at Bob's side in the quantum state merging task \cite{Horodecki2009entanglement}. A negative value indicates that one does not need to send qubits, but can also \emph{gain} $-S(A|B)$ maximally entangled pairs. Considering the two states given before, suppose one can transform $\rho$ to $\psi_\mathcal{G}$ with local isometries, and we show this is contradict to the quantum state merging efficiency. First, transform $\rho$ to $\psi_\mathcal{G}$ with local isometries, and then one can finish the quantum state merging of $\psi_\mathcal{G}$ without any qubit sending but get one entangled pair. Finally, one can transform $\psi_\mathcal{G}$ back to $\rho$ with local operations according to Eq.~\eqref{Eq:Ap:LU}, which contradicts to $S(A|B)_{\rho}\geq 0$.

\section{Proof for Applications}

\emph{Proof for Corollary \ref{Cor:constant}.---}
We divide the $N$ vertices of the graph with chromatic number $K$ into $K$ disjointed subset $C_k$. 
According to Lemma \ref{Le:App} below, one can always find a vertex, without loss of generality, denoted as the vertex $1$ belonging to the first color subset, $1\in C_1$ satisfying that 
from every other color subset $C_k$ $(k\ne 1)$, there always exists at least one vertex, $v_k\in C_k$ such that vertex $1$ and $v_k$ are neighbors, i.e., $v_k\in n_1$.
We construct $K$ disjoint stabilizer sets as follows,
\begin{equation}
\begin{aligned}
    &\mathcal{P}_1=\{G_1,G_1G_j| j\in C_1\setminus \{1\} \},\\ &\mathcal{P}_k=\{G_i,G_kG_j| i\in C_k \cap n_1, j\in C_k\setminus n_1\}, k={2,3\dots,K}.\\
\end{aligned}
\end{equation}
where $G_i$ denotes the generator of the $i$th vertex. Here $\mathcal{P}_1$ contains $G_1$ and the multiplications of $G_1$ with the other generators from the first color set $C_1$. For $2\leq k \leq K$, $\mathcal{P}_k$ contains the generators of the neighbors of vertex $1$, and the multiplications of these generators with the other generators from the same color set $C_k$. For simplicity of the construction, we only consider the multiplication of $G_j$ with one of the generators from $n_1$, say $G_k$.

It is not hard to check that any pair of 
$S^i\in \mathcal{P}_1,S^j \in \mathcal{P}_k$ are \emph{pairable}. In particular, they are anti-commutative at the first position $i=1$.
We denote the stabilizers in $\mathcal{P}_1$ as $S^1,\dots S^{|C_1|}$ and other stabilizers in $\mathcal{P}_2,\dots,\mathcal{P}_K$ as $S^{|C_1|+1},\dots S^N$. We choose $\mathcal{AC}=\{1\}$ , and when $2|C_1|\le N $ we construct the corresponding Bell inequality by reusing stabilizer $S^1$ as follows,
\begin{equation}
\begin{aligned}
(N-2|C_1|+1)\mathcal{B}_1+ \sum_{j=2}^{N}\mathcal{B}_j\le \beta_C=2(N-|C_1|)
 \end{aligned}
\end{equation}
where $\mathcal{B}_{j}=\left \langle (A_1+s^j_1B_1)\prod_{i\ne 1}P_i({s_i^j})\right \rangle $. 
When $2|C_1|> N $, we reuse stabilizer $S^N$, and construct Bell inequality
\begin{equation}
\begin{aligned}
(2|C_1|-N+1)\mathcal{B}_N+ \sum_{j=1}^{N-1}\mathcal{B}_j\le \beta_C=2|C_1|
 \end{aligned}
\end{equation}
These two Bell inequality both has $2N$ correlations and the maximal quantum values are $\beta_Q=2\sqrt{2}(N-|C_1|)$, $\beta_Q=2\sqrt{2}|C_1|$, respectively. As the result, the ratio $\beta_Q/\beta_C=\sqrt{2}$ for both cases.

The stabilizers we used are $\mathcal{ST}=\bigcup_{k=1,2\dots,K} \mathcal{P}_k$  containing $N$ independent stabilizers. Via Theorem 2, we prove that the above inequality can self-test the graph state $\psi_{\mathcal{G}}$.

\begin{lemma}\label{Le:App}
For any graph $\mathcal{G}=(V,E)$ whose chromatic number is $K$ with disjoint color sets $C_1,\dots, C_K$, $V=\bigcup_{k=1}^K C_k$, there always exists a vertex, for example, $1\in C_1$ satisfying that from every other color subset, 
we could find at least one vertex $v_k\in C_k$ $k=2,\dots,K$ such that vertex $1$ and $v_k$ are neighbors, $v_k\in n_1$.
\end{lemma}

\begin{proof}
We prove this lemma by contradiction. Assuming that we could not find the vertex satisfying the requirement, we denote the vertex with the maximal different color neighbors as vertex 1. Thus there exists a color set, for example $C_K$, no vertex from it is neighbor to vertex 1. Then we pick one vertex from $C_K$, denoted as vertex $v_K$, 
we can always find a color set $C_q$ $(q\ne K)$ such that every vertex in $C_q$ is not neighbor to $v_K$, otherwise it is contradictive to the assumption that vertex 1 has the maximal different color neighbors. Then we can color this vertex $v_K$ with color $q$. Repeat this procedure for the vertices in color set $C_K$ and we find that this graph can be colored with $K-1$ color, which causes a contradiction. 
\end{proof}

\section{Robust self-test of graph state}
In this section, we give the detailed explanation about numerical robustness results of self-testing shown in FIG \ref{Fig:compare}. We would like to lower bound the fidelity between the measured state $\rho$ and the target graph state $\psi_G$ (under local isometry), with the knowledge of the Bell inequality value. To this end, mathematical equivalently, one can adopt local extraction channel and the maximal fidelity shows
\begin{equation}
\begin{aligned}
F=\max_{\Lambda=\Lambda_1\otimes\Lambda_2\cdots\Lambda_N } \bra{\psi_G} \Lambda(\rho) \ket{\psi_G},
\end{aligned}
\end{equation}
where $\Lambda_i$ is the local channel on $i$-th party. Alternatively, the fidelity can be written as follows,
\begin{equation}
\begin{aligned}
\Tr[\rho\Lambda_1^\dag \otimes\Lambda_2^\dag \cdots\Lambda_N^\dag (\ket{\psi_G}\bra{\psi_G})]
\end{aligned}
\end{equation}
where $\Lambda_i^\dag$ is the dual channel of $\Lambda_i$. Note that here the dual of the extraction map acts on the graph state, and we denote the state after this dual channel as $K=\Lambda_1^\dag \otimes\Lambda_2^\dag \cdots\Lambda_N^\dag (\ket{\psi_G}\bra{\psi_G})$.

To find a reliable lower bound of the fidelity from the Bell inequality value, one can choose appropriate parameters $s$ and $\mu$ such that the following inequality on operators always holds,
\begin{equation}\label{Eq:Lbound}
\begin{aligned}
K\geq s\mathcal{B}+\mu \mathbb{I}
\end{aligned}
\end{equation}
where $\mathcal{B}$ is Bell inequality to self-test the state. In this way, the fidelity is bounded as $F\geq s\beta+\mu$, with $\beta=\Tr(\rho\mathcal{B})$ the Bell inequality value. 

Since the measurement of $\mathcal{B}$ is restricted to the dichotomic scenario, on account of the Jordan lemma, one can reduce the state to the N-qubit space, and the possible measurements can be parameterized by  the angles $\theta_i\in[0,\pi/2]$ as,
\begin{equation}
\begin{aligned}
A_i&=\cos \theta_i X_i + \sin \theta_i Z_i,\\
B_i&=\cos \theta_i X_i - \sin \theta_i Z_i
\end{aligned}
\end{equation}
for $i\in T$, the rotated set, and for other qubits
\begin{equation}
\begin{aligned}
A_i&=\cos \theta_i H_i + \sin \theta_i V_i,\\
B_i&=\cos \theta_i H_i - \sin \theta_i V_i
\end{aligned}
\end{equation}
where $H_i(V_i)=(X_i\pm Z_i)/\sqrt{2}$. Consider a specific extraction channel in Ref.~\cite{kaniewski2016analytic}, 
\begin{equation}
\begin{aligned}
\Lambda_i(\rho)=\frac{1+g(x)}{2}\rho+\frac{1-g(x)}{2}\Gamma_i(x)(\rho)\Gamma_i(x),
\end{aligned}
\end{equation}
where $g(x)=(1+\sqrt{2})(\sin x+\cos x+1)$, and  $\Gamma_i(x)$ is the operator on $i$-th qubit: for $i\in T$, $\Gamma_i(x)=X_i(Z_i)$ as $x< (\geq) \pi/4$; for $i\notin T$, $\Gamma_i(x)=H_i(V_i)$ as $x< (\geq) \pi/4$. Now the Bell inequality $\mathcal{B}$ and the operator $K$ are both parameterized with $\theta_i$, and the inequality in Eq.~\eqref{Eq:Lbound} shows, 
\begin{equation}\label{Eq:Lbound1}
\begin{aligned}
K(\theta_1, \theta_2, \cdots, \theta_N) \geq s\mathcal{B}(\theta_1, \theta_2, \cdots, \theta_N)+\mu \mathbb{I}.
\end{aligned}
\end{equation}
where we should find an optimal $s$ and $\mu$ for all possible $\theta_1, \theta_2, \cdots, \theta_N$.

In FIG.~\ref{Fig:compare}, all the robustness results for 3, 4-partite GHZ or cluster states are obtained from the above inequality numerically. To be specific, given a fixed $s$, we find the minima of the minimal eigenvalue of $K(\vec{\theta})-s\mathcal{B}(\vec{\theta})$ for all $\vec{\theta}$. 
A slower slope indicates a better bound. Thus, to find the optimal linear bound, we let the relation $s\beta_Q+\mu=1$ hold, that is, the fidelity approaches $1$ for the maximal quantum value and find the minimum $s$. 
We list all the obtained $s$ and $\mu$ in FIG.~\ref{Fig:compare} as follows.

\begin{table}[ht]
\begin{tabular}{|cl|cc|}
  
  \hline
  &3-qubit GHZ (cluster) & $s$ & $\mu$ \\[1mm]
  \hline
 & $\mathbf{1}$ \& $\mathbf{2}$  & 0.906
                        & -2.4686
                            \\[1mm]
            \hline
 & $\mathbf{3}$  \& $\mathbf{4}$   & 0.6036
                        & -2.4145
                              \\[1mm]
  \hline
    & Mermin \cite{kaniewski2016analytic} & $\frac{2+\sqrt{2}}{8}$
                        & $-\frac{1}{\sqrt{2}}$
                           \\[1mm]

  \hline
\end{tabular}
\caption{Numerical fidelity bound for the 3-qubit GHZ (cluster) state.   }
\label{Table:3GHZresults}
\end{table}

\begin{table}[ht]
\begin{tabular}{|cl|cc|}
  
  \hline
  &4-qubit GHZ  & $s$ & $\mu$ \\[1mm]
  \hline
 & $\mathbf{1}$ \& $\mathbf{2}$  & 1
                        &$-1-2\sqrt{2}$
                            \\[1mm]
            \hline
 & $\mathbf{3}$    & 0.69
                        & -3.5931
                              \\[1mm]
  
 \hline
  & $\mathbf{4}$   & 0.49
                        & -3.1578
                              \\[1mm]
  
  \hline
    & Mermin  & 0.219
                        & -0.752
                           \\[1mm]

  \hline
\end{tabular}
\caption{Numerical fidelity bound for the 4-qubit GHZ state.   }
\end{table}

\begin{table}[ht]
\begin{tabular}{|cl|cc|}
  \hline
  &4-qubit cluster & $s$ & $\mu$ \\[1mm]
  \hline
 & $\mathbf{1}$  & 1
                        & $-1-2\sqrt{2}$
                            \\[1mm]
     \hline
  & $\mathbf{2}$   & 0.7400
                        & -3.9262
                            \\[1mm]                          
            \hline
 & $\mathbf{3}$  \& $\mathbf{4}$   & 0.6200
                        & -2.5071
                              \\[1mm]
  \hline
\end{tabular}
\caption{Numerical fidelity bound for the 4-qubit cluster state.   }
\end{table}

\section{Device-independent Entanglement Witness}
As a side result of the robust self-testing bound, applying the linear self-testing bound shown above, we can also construct device-independent entanglement witness.
In the following, we give genuine entanglement bounds $\beta_{0.5}$ and the detailed construction for \emph{Single pair} inequality with the best known robustness bound. The device-independent genuine entanglement witness of 3-party, 4-party GHZ and 4-party cluster are shown, respectively. Any violation of these inequalities implies the existence of genuine entanglement and the similar method can also be used to detect more detailed entanglement structures in a device-independent manner. 
\begin{widetext}
   \begin{equation}
  \begin{aligned}
  &\big \langle(A_1+B_1)B_2 \big \rangle+ \big \langle(A_1-B_1)A_2B_3 \big\rangle +
   \big\langle B_2A_3 \big\rangle
  \stackrel{\text{bi-sep.}}{\le} 3.2766 \\
  & \big\langle(A_1+B_1)B_2B_3B_4 \big \rangle+  \big\langle(A_1-B_1)A_2 \big\rangle + \big\langle A_2A_3  \big\rangle+ \big\langle A_2A_4 \big \rangle
  \stackrel{\text{bi-sep.}}{\le} \frac{3}{2}+2\sqrt{2} \\
   &\big  \langle(A_1+B_1)B_2\big  \rangle+ \big  \langle(A_1-B_1)A_2B_3  \big  \rangle +\big  \langle B_2A_3B_4 \big \rangle+\big  \langle B_3A_4 \big  \rangle
  \stackrel{\text{bi-sep.}}{\le} \frac{3}{2}+2\sqrt{2}  \\
 \end{aligned}
\end{equation}  
\end{widetext}

\end{appendix}
\end{document}